\documentclass{article}
\usepackage{ijcai26}

\usepackage{times}
\usepackage{soul}
\usepackage{url}
\usepackage[hidelinks]{hyperref}
\usepackage[utf8]{inputenc}
\usepackage[small]{caption}
\usepackage{graphicx}
\usepackage{amsmath}
\usepackage{amsthm}
\usepackage{booktabs}
\usepackage{algorithm}
\usepackage{algorithmic}
\usepackage[switch]{lineno}
\usepackage{amssymb}
\usepackage{xcolor}
\usepackage{multirow}
\usepackage{tcolorbox}

\newtheorem{theorem}{Theorem}
\newtheorem{definition}{Definition}
\newtheorem{lemma}{Lemma}

\title{A Faster Deterministic Algorithm for Kidney Exchange via Representative Set}

\author{
Kangyi Tian
\And
Mingyu Xiao\thanks{Corresponding Author}\\
\affiliations
School of Computer Science and Engineering, University of Electronic Science and Technology of China, Chengdu, China\\
\emails
kangyitian947@gmail.com,
myxiao@gmail.com
}

\begin{document}

\maketitle

\begin{abstract}
The Kidney Exchange Problem is a prominent challenge in healthcare and economics, arising in the context of organ transplantation. It has been extensively studied in artificial intelligence and optimization. In a kidney exchange, a set of donor-recipient pairs and altruistic donors are considered, with the goal of identifying a sequence of exchanges—comprising cycles or chains starting from altruistic donors—such that each donor provides a kidney to the compatible recipient in the next donor-recipient pair. 
Due to constraints in medical resources, some limits are often imposed on the lengths of these cycles and chains.
These exchanges create a network of transplants aimed at maximizing the total number, $t$, of successful transplants.  Recently, this problem was deterministically solved in $O^*(14.34^t)$ time (IJCAI 2024). In this paper, we introduce the representative set technique for the Kidney Exchange Problem, showing that the problem can be deterministically solved in $O^*(6.855^t)$ time.
\end{abstract}

\section{Introduction}

Kidney disease has become an increasingly significant global issue. According to international consensus~\cite{francis2024chronic}, approximately 850 million people worldwide are estimated to suffer from kidney disease. Patients with kidney disease are typically treated through either \emph{dialysis} or \emph{kidney transplantation}. In terms of therapeutic outcomes, kidney transplantation is generally preferred over dialysis for treating kidney disease. However, there is a substantial gap between the number of kidney donors and recipients. Patients unable to find a compatible donor for a transplant are placed on a waiting list. Unfortunately, as the waiting list grows longer each year, many patients miss the optimal window for treatment and may even lose their lives while waiting. 

Often, patients on the waiting list have family members or friends willing to donate a kidney, but the kidney may not be compatible. To address this challenge, Kidney Paired Donation (KPD), also known as Kidney Exchange, was introduced by Rapaport~\cite{Rapaport86}. Since its introduction in 1986, KPD has been implemented in several countries and has helped many people receive compatible kidneys~\cite{pahwa2012paired,DBLP:conf/sigecom/DickersonMPST16}.

In Kidney Paired Donation, each participant is either an altruistic donor or a donor-recipient pair. To ensure that every recipient receives a compatible kidney, kidney donations are organized into \emph{chain models} or \emph{cycle models}. In chain models, an altruistic donor first donates a kidney to the next donor-recipient pair. Then, each subsequent donor-recipient pair in the chain receives a kidney from the previous pair and donates a kidney to the next pair, if one exists. In cycle models, each pair receives a kidney from the previous pair and donates a kidney to the next pair. It is important to note that in both models, a pair can choose to exit the program after receiving a kidney, as they have no legal obligation to donate. From a fairness perspective, we aim to avoid this, particularly in cycle models. To prevent this issue, all kidney transplant surgeries are ideally performed simultaneously. However, due to medical and logistical constraints, there can only be a limited number of surgeries at one time. As a result, the length of chains and cycles cannot be excessively long.\footnote{The length restriction for chains is generally less stringent than for cycles. Some studies have considered versions with more relaxed chain length restrictions~\cite{Anderson2015FindingLC,DBLP:journals/msom/GlorieKW14}.}

To optimize the number of kidney exchanges, the central problem in Kidney Paired Donation can be framed as follows. First, a compatibility graph is constructed based on medical compatibility criteria. In this graph, each vertex represents either an altruistic donor or a donor-recipient pair, and a directed edge exists between two vertices if one can donate a compatible kidney to the other. We can assume that self-loops do not exist in the compatibility graph. Our computational task then becomes the problem of finding the largest vertex-disjoint packing of cycles and paths starting from an altruistic vertex in the compatibility graph, subject to restrictions on the length of paths and cycles. This problem is referred to as the Kidney Exchange Problem.

We now present the formal definition of the Kidney Exchange Problem. Let $G$ be the compatibility graph, and $B$ be the set of altruistic vertices. We are given two nonnegative integers, $l_p$ and $l_c$. The \textit{length} of a path or cycle is defined as the number of edges it contains. A path is \emph{feasible} if it starts at an altruistic vertex in $B$ and has a length of at most $l_p$, while a cycle is \textit{feasible} if it does not contain any vertices in $B$ and its length does not exceed $l_c$.
A path-cycle packing is a set of vertex-disjoint cycles and paths. A path-cycle packing is \textit{feasible} if all cycles and paths in it are feasible. 

\begin{tcolorbox}[colframe=black, colback=white]
\textbf{Kidney Exchange Problem (KEP)}\\
\textbf{Input:} A directed graph $G = (V,A)$ without self-loops, an altruistic vertex set $B \subseteq V$ without any edge entering a vertex in $B$, two nonnegative integers $l_p$ and $l_c$, and a target integer $t$.\\
\textbf{Output:} A feasible path-cycle packing such that the total length of cycles and paths in it is at least $t$.
\end{tcolorbox}

An instance of KEP is denoted by \( (G, B, l_p, l_c, t) \) .

\subsection{Related Work}

The Kidney Exchange Problem (KEP) has been extensively studied since its introduction in \cite{Rapaport86}. In most studies, KEP is typically modeled as a packing problem involving vertex-disjoint chains and cycles. Due to variations in the types of packings and constraints on the lengths of chains and cycles, many variants have been proposed in recent years \cite{DBLP:journals/eor/ConstantinoKVR13,DBLP:journals/jea/ManloveO14,DBLP:journals/msom/GlorieKW14,Anderson2015FindingLC}.

Numerous practical algorithms have been developed to address the Kidney Exchange Problem. A traditional approach to solving KEP is through Integer Programming. Based on Integer Programming, \cite{DBLP:journals/jea/ManloveO14} designed an algorithm and deployed it in software, which has been used to identify kidney exchange pairs in real-world applications. More recently, a new method using graph neural networks was introduced to KEP by \cite{DBLP:conf/ictai/PimentaAL23}.

Theoretical studies have shown that KEP is NP-hard, even for some very restricted versions \cite{DBLP:journals/talg/KrivelevichNSYY07,DBLP:conf/sigecom/AbrahamBS07,DBLP:journals/dmaa/BiroMR09}. Due to the NP-hardness of KEP, approximation algorithms \cite{DBLP:journals/talg/KrivelevichNSYY07,DBLP:conf/atal/JiaTWZ17}, randomized algorithms \cite{DBLP:journals/algorithms/LinWFF19}, and parameterized algorithms \cite{DBLP:conf/ijcai/XiaoW18,DBLP:conf/ijcai/MaitiD22,DBLP:conf/ijcai/Hebert-JohnsonL24} have been introduced to handle its computational complexity.

We focus on  KEP parameterized by the number of covered recipients $t$. Although this problem has been studied for many years, whether it is fixed-parameter tractable (FPT) with respect to~$t$ remained unresolved until recently.

\cite{DBLP:conf/ijcai/MaitiD22} was the first  work to show that KEP is FPT by employing Color-Coding. Their randomized algorithm runs in time $O^*(8^t)$  and can be derandomized to run in $O^*(161^t)$. More recently, using the technique of arithmetic circuit representation of polynomials~\cite{DBLP:journals/ipl/Williams09}, \cite{DBLP:conf/ijcai/Hebert-JohnsonL24} developed a randomized algorithm with an improved running time of $O^*(4^t)$. After derandomization, the running time becomes $O^*(14.34^t)$. 
Very recently, Banik et al.~\cite{banik2025kidneyexchangefasterparameterized} proposed a simple algorithm with a claimed runtime of $O^*(10.88^{t})$. 
However, their proof is incomplete, and we observe that their algorithm in fact requires $O^*(10.88^{2t})$ time. 
The discrepancy arises from their use of the color-coding technique introduced for KEP in \cite{DBLP:conf/ijcai/MaitiD22}, where they assign $t$ colors to the vertices. 
This is insufficient, because a feasible solution may contain up to $2t-2$ vertices---for instance, two disjoint cycles each of length $t-1$. 
Thus, one must instead use roughly $2t$ colors, which increases the exponential dependency from $10.88^{t}$ to $10.88^{2t}$. 
A similar issue explains why the randomized algorithm in \cite{DBLP:conf/ijcai/Hebert-JohnsonL24} 
runs in $O^{*}(4^{t})$ time rather than $O^{*}(2^{t})$.


\subsection{Our Contribution and Method}

The main contribution of this paper is to introduce the \emph{representative set} technique to KEP, leading to a \emph{deterministic algorithm with running time $O^*(6.855^t)$}. This improves upon the previous best deterministic bound of $O^*(14.34^t)$~\cite{DBLP:conf/ijcai/Hebert-JohnsonL24}.

In fact, by applying our Lemma~\ref{lem-2t}, we can directly improve the running time of \citeauthor{DBLP:conf/ijcai/MaitiD22}'s randomized algorithm to $O^*(4^t)$.
See Theorem~\ref{thm-3} in Section~\ref{sec-random} for the detailed proof.
However, straightforward derandomization leads to a time bound of $O^*(30^t)$. Achieving better deterministic performance requires more advanced techniques.

From a technical perspective, our approach begins with a dynamic programming (DP) algorithm that step-by-step constructs all feasible path-cycle packings of total length at most $t$. However, a direct implementation faces a major bottleneck: the number of intermediate packings can be exponential and is not bounded by a function of $t$. Consequently, this approach does not yield a truly parameterized algorithm.

To address this, we compress the set of feasible path-cycle packings at each DP state. The key insight is that we only care about the set of recipients covered by the paths and cycles, not their specific structure. Thus, one set of cycles and paths can be replaced with another as long as they cover the same vertices.

To enable this compression, we employ the \emph{representative set} technique~\cite{DBLP:journals/jacm/FominLPS16-rep}. A representative set is a carefully selected subset that captures the essential combinatorial properties of the larger family. In our setting, it enables us to preserve optimality while significantly reducing the search space.

Although representative sets have been well studied~\cite{DBLP:journals/jacm/FominLPS16-rep}, existing methods cannot be directly applied to KEP due to the presence of local constraints and problem-specific structural features. In particular, our dynamic programming must be carefully designed to ensure that no valid solutions are lost during the compression step. As a result, our correctness proof and running time analysis differ significantly from those of earlier applications of the representative set method.

\section{Preliminaries}
Let $G=(V,A)$ denote a directed graph. We will let $n=|V|$.
The graph $G$ will be used to denote the the compatibility graph in KEP.
We will always use $B$ to represent the set of all altruistic donors. All vertices in $B$ are sources, i.e., no edge enters a vertex in $B$.

A \emph{path} of length $k-1$ is sequence of distinctive vertices
$v_1 v_2 \cdots v_k$ such that \((v_i, v_{i+1}) \in A\) for all \(i = 1, 2, \cdots, k-1\).
For a path $v_1 v_2 \cdots v_k$, if there is also an edge \((v_k, v_{1}) \in A\), we call $v_1 v_2 \cdots v_k v_1$ a \emph{cycle} of length $k$.
For a path $P=v_1v_2\cdots v_k$ (resp., a cycle $C=v_1v_2\cdots v_k v_1$), we let $V(P)=\{v_1,v_2,\cdots,v_k\}$ (resp., $V(C)=\{v_1,v_2,\cdots,v_k\}$).
For a path-cycle packing $\mathcal{D}$, we define the \emph{size} $|\mathcal{D}|$ as the number of cycles and paths in $\mathcal{D}$.
We also define the \emph{length} $c(\mathcal{D})$ as the sum of lengths of all paths and cycles in $\mathcal{D}$. Note that $c(\mathcal{D})$ may be smaller than the number of vertices appearing in $\mathcal{D}$.
A singleton $\{D\}$ may be simply written as $D$.


Let  $\mathcal{F}$ be a collection of sets and $S$ be another set. we define $\mathcal{F} \uplus S = \{X \cup S| X\in \mathcal{F}, X \cap S = \emptyset \}$.

Let $\mathcal{S}$ be a family of sets and the cardinality $|\mathcal{S}|$ denote the number of sets in $\mathcal{S}$. We say that a family is \emph{$p$-family} if $|X|=p$ holds for all sets $X \in \mathcal{S}$. 

We use $O^*(f(k))$ to denote $f(k)n^{O(1)}$.

\paragraph{Representative Set.}
Representative set is one of the main techniques used in this paper.
Usually,  representative sets are defined using matroids~\cite{DBLP:journals/jacm/FominLPS16-rep}. However, in this paper, we only use the technique in a simpler situation. So, we define representative sets in a simpler way without using the concept of matroids.

\begin{definition}[$q$-Representative Set]
Given a ground set $E$ and a $p$-family $\mathcal{S}$ of subsets of $E$.
A subfamily $\widehat{\mathcal{S}} \subseteq \mathcal{S}$ is \emph{$q$-representative} for $\mathcal{S}$ if the following holds:
for every set $Y \subseteq E$ of size at most $q$, if there is a set $X \in \mathcal{S}$ disjoint from $Y$,
then there is a set $\widehat{X} \in \widehat{\mathcal{S}}$ disjoint from $Y$.
We write $\widehat{\mathcal{S}} \subseteq^q_{\text{rep}} \mathcal{S}$ to denote that $\widehat{\mathcal{S}} \subseteq \mathcal{S}$ is $q$-representative for $\mathcal{S}$.
\label{def-rep-family}
\end{definition}

In the definition, we require that $\mathcal{S}$ to be a $p$-family. The requirement on the size seems not used in the definition. However, we specify the information because this is related to the running time and the size of the representative set. 

\cite{DBLP:journals/jacm/FominLPS16-rep} proposed the following lemmas to compute the $q$-representative sets efficiently. We will also use them.

\begin{lemma}[\cite{DBLP:journals/jacm/FominLPS16-rep}]
    Suppose that $\mathcal{X}$, $\mathcal{Y}$ and $\mathcal{Z}$ are p-families.
    If $\mathcal{X} \subseteq^q_{\text{rep}} \mathcal{Y}$ and $\mathcal{Y} \subseteq^q_{\text{rep}} \mathcal{Z}$, then $\mathcal{X} \subseteq^q_{\text{rep}} \mathcal{Z}$.
    \label{lem-rep-trans}
\end{lemma}

\begin{lemma}[\cite{DBLP:journals/jacm/FominLPS16-rep}]
     There is an algorithm that, given a $p$-family $\mathcal{S}$ of sets over a ground $E$ of size $n$, an integer $q$,
     computes in time
    \[
    \mathcal{O} \big( 
    |\mathcal{S}| \cdot (1-x)^{-q} \cdot 2^{o(p+q)} \cdot \log n \big
    )
    \]
    a subfamily $\widehat{\mathcal{S}} \subseteq \mathcal{S}$ such that $|\widehat{\mathcal{S}}| \leq x^{-p}(1-x)^{-q} \cdot 2^{o(p+q)}$ and $\widehat{\mathcal{S}} \subseteq_{rep}^q \mathcal{S} .$
    \label{lem-rep-family}
\end{lemma}

\section{A Faster Deterministic Algorithm}
\label{sec-FPT}

In this section, we will introduce our algorithm. We are going to prove the following main result. 
\begin{theorem}
KEP can be deterministically solved in $O^*(6.855^{t})$ time.
    \label{thm-kep}
\end{theorem}

In our main algorithm, we first handle the case when $l_p \ge t$ or $l_c \ge t$.
For this case, we check whether there is a feasible path of length $t$ or a feasible cycle of length at least $t$. If so, we find a solution covering at least $t$ recipients and the problem is solved.
If no, we can set that $l_p = \min(l_p, t-1)$ and $l_c = \min(l_c,t-1)$.
Finding feasible paths and cycles of length at least $t$
can be done by slightly modifying the known algorithms of $k$-Path and Long Directed Cycle in~\cite{DBLP:journals/jacm/FominLPS16-rep}.
To find a feasible $k$-path, we only need to modify the algorithm for $k$-Path by restricting starting points of paths in $B$ in the dynamic programming. For a feasible cycle of length at least $t$,
we can use their algorithm to find the smallest cycles of length at least $t$.
Thus, we can check whether there is a feasible path (resp., a feasible cycle) of length $\ge t$ in time $O(2.619^tn \log ^2n)$ (resp., $O(6.75^{t+o(t)}mn^2)$) according to~\cite{DBLP:journals/jacm/FominLPS16-rep}.

Next, we always assume that $l_p < t$ and $l_c <t$.
The algorithm will contain three major parts.
In Section~\ref{sec-semi},
we first introduce the concept of \textit{semi-feasible path-cycle packing},
which is an important concept used in the middle steps of our algorithm. In Section~\ref{sec-dp}, we introduce a dynamic programming algorithm to compute semi-feasible path-cycle packings.
This algorithm will solve KEP with $l_p < t$ and $l_c <t$ in $O^*(2^n)$ time.
After that, in Section~\ref{sec-rep}, we show how to apply representative sets to compress the number of  semi-feasible path-cycle packings in the dynamic programming algorithm. As a result, we will obtain an efficient algorithm for KEP.

\subsection{Semi-feasible path-cycle Packing}
\label{sec-semi}
Our idea is to use a dynamic programming algorithm to compute feasible path-cycle packings step by step. However, in some meddle steps, we need to store a partial of a feasible path or cycle. Thus, we come up with the concept of semi-feasible path-cycle packing.

A path-cycle packing $\mathcal{D}=\mathcal{P} \cup \mathcal{C} \cup \{D\}$ is \emph{semi-feasible} if it satisfies the following conditions:
\begin{itemize}
    \item $\mathcal{P}$ is a feasible path packing;
    \item $\mathcal{C}$ is a feasible cycle packing;
    \item $D$ is either a feasible path or a feasible cycle or a path of length at most $l_c$ such that $V(D) \cap B=\emptyset$.
\end{itemize}

We use the example in Figure~\ref{fig-semi} illustrate semi-feasible path-cycle packings.
Let $P_1 = abc$, $P_2 = fgh$, $P_3=igh$, $C_1 = ded$, and $C_2 = ghig$.
Since $l_p=l_c=3$, we know that $P_1$, $P_2$, $C_1$, and $C_2$ are all feasible.
Then the path-cycle packing $\mathcal{D}_1=\{P_1,C_1,P_2\}$ is both semi-feasible and feasible.
The path-cycle packing $\mathcal{D}_2=\{P_1\}\cup \{C_1\}\cup \{P_3\}$ is not a feasible packing since $P_3$ is not feasible.
However, $\mathcal{D}_2$ is a semi-feasible path-cycle packing since $P_3$ is of length $\le l_c$ and $V(P_3) \cap B=\emptyset$.

\begin{figure}[t]
    \centering
    \includegraphics[scale = 0.18]{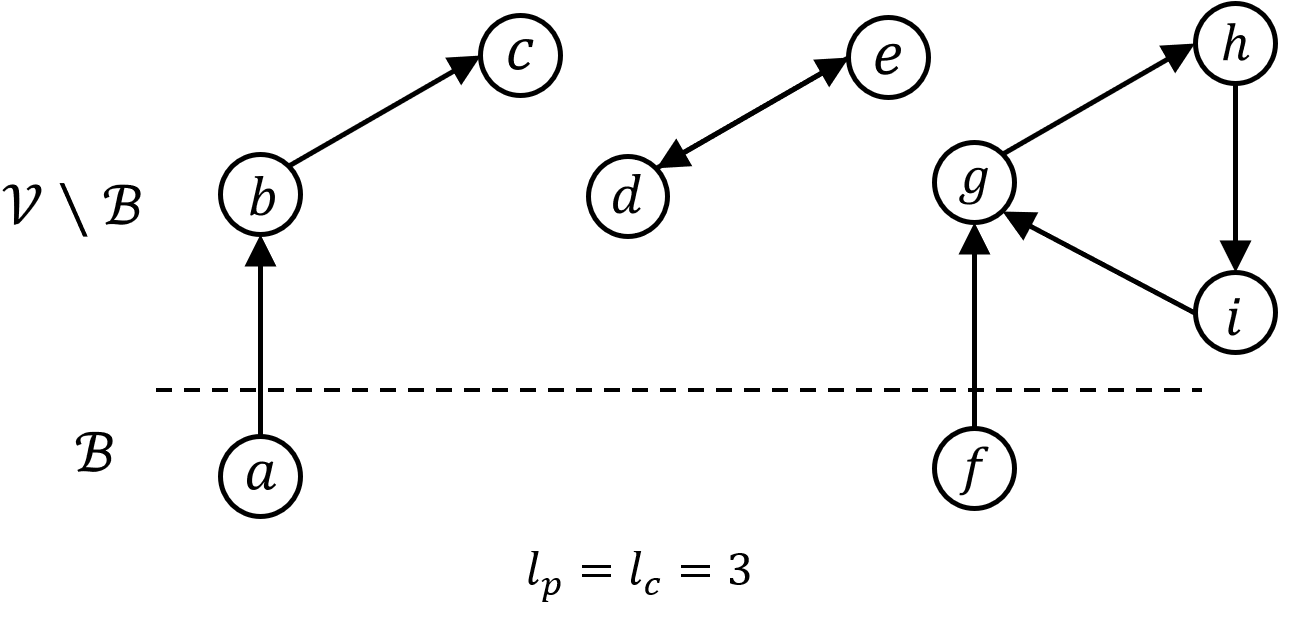}
    \centering
    \centering
        \caption{A compatibility graph}
    \label{fig-semi}
\end{figure}

\subsection{A Dynamic Programming Algorithm}
\label{sec-dp}

In this section, we are going to use a dynamic programming method
to solve KEP with $l_p < t$ and $l_c <t$ in $O^*(2^n)$ time. This running time bound is not new. Previously, \cite{DBLP:conf/ijcai/XiaoW18} showed that KEP can be solved in the same running time bound by using subset convolution. However, their approach cannot combine the representative set method to get the result in this paper. Next, we prove the following theorem.

\begin{theorem}
There is a dynamic programming algorithm that can deterministically solve KEP with $l_p < t$ and $l_c < t$ in $O^*(2^n)$ time.

    \label{thm-dp}
\end{theorem}



We use KEPalg to denote the dynamic programming algorithm in Theorem~\ref{thm-dp}.
We also recall the assumption that $l_p < t$ and $l_c <t$.
Under this assumption, we have the following property to bound the length of a  feasible path-cycle packing.

\begin{lemma}
    Assume that $l_p < t$ and $l_c < t$. A KEP instance $(G,B,l_p,l_c,t)$ is a yes-instance if and only if there exists a feasible path-cycle packing $\mathcal{D}$ such that $t \le c(\mathcal{D}) \le 2t$.
     \label{lem-2t}
 \end{lemma}
\begin{proof}
    If there is a feasible path-cycle packing $\mathcal{D}$ such that $t \le c(\mathcal{D}) \le 2t$, 
    then the instance is a yes-instance by the definition of the problem.
    
    Now we prove the other direction.
    If $(G,B,l_p,l_c,t)$ is a yes-instance, then there exists a feasible cycle-path packing $\mathcal{D}$ such that $c(\mathcal{D}) \ge t$. We assume that $\mathcal{D}$ is the packing with the minimum $c(\mathcal{D})$ satisfying the above requirement.
    We are done if it also holds that $c(\mathcal{D}) \le 2t$.
    Otherwise, we have that $c(D) > 2t$. Since $l_c < t$ and $l_p < t$, it holds that $c(D) < t$ for any path or cycle $D \in \mathcal{D}$. 
    Let $\mathcal{D}'=\mathcal{D} \setminus D$. since $c(\mathcal{D}') = c(\mathcal{D}) - c(D) \ge t$, we obtain a new feasible path-cycle packing $\mathcal{D}'$ such that $t \le c(\mathcal{D'}) < c(\mathcal{D})$, which is a contradiction to the minimality of $\mathcal{D}$. So, we know that it holds $t \le c(\mathcal{D}) \le 2t$.
\end{proof}

This lemma says that there exists a path-cycle packing of length at most $2t$ if the input is a yes-instance. In the process of our dynamic programming, we will compute the vertex subsets of all possible semi-feasible packings with length at most $2t$.
Finally, we answer the instance by checking whether there is a feasible path-cycle packing of length at least $t$ among all semi-feasible packings we compute.

\begin{table}[t]
    \small
    \centering
    \begin{tabular}{c|c|c|c|c}
        \hline
        \multirow{3}{*}{} & \multicolumn{3}{c|}{\textbf{Semi-feasible Packing}}  & \textbf{Vertex Subset}\\
        & \multicolumn{3}{c|}{$\mathcal{D} = \mathcal{P} \cup \mathcal{C} \cup \{D\}$} & \multirow{2}{*}{$V(\mathcal{D})$} \\
        \cline{2-4}
        & $\mathcal{P}$ & $\mathcal{C}$ & $D$ \\
        \hline
        $V(\mathcal{D}) \in \mathcal{F}^{ab}_{1,1}$ & $\emptyset$ & $\emptyset$ & $ab$ & $\{a, b\} $ \\ 
        \hline
        $ V(\mathcal{D})\in \mathcal{F}^{ac}_{2,2}$ & $\emptyset$ & $\emptyset$ & $abc$ & $\{a, b, c\}$ \\ 
        \hline
        $V(\mathcal{D}) \in \mathcal{F}^{de}_{3,1}$ & $\{abc\}$ & $\emptyset$ & $de$ & \multirow{2}{*}{$\{a, b, c, d, e\}$} \\ 
        \cline{1-4}
        $V(\mathcal{D}) \in \mathcal{F}^{ed}_{3,1}$ & $\{abc\}$ & $\emptyset$ & $ed$ &  \\ 
        \hline
        \multirow{2}{*}{$V(\mathcal{D}) \in \mathcal{F}^{fg}_{3,1}$} & $\{abc\}$ & $\emptyset$ & $fg$ & $\{a, b, c, f, g\}$ \\ 
        \cline{2-5}
         & $\emptyset$ & $\{ded\}$ & $fg$ & $\{d, e, f, g\}$ \\ 
        \hline
        $V(\mathcal{D}) \in \mathcal{F}^{dd}_{4,2}$ & $\{abc\}$ & $\emptyset$ & $ded$ & $\{a, b, c, d, e\}$ \\ 
        \hline
        \textbf{$V(\mathcal{D}) \in \mathcal{F}^{fg}_{5,1}$} & $\{abc\}$ & $\{ded\}$ & $fg$ & $\{a, b, c, d, e, f, g\}$ \\ 
        \hline
    \end{tabular}

     \caption{Some illustrations of $\mathcal{F}_{k,l}^{vu}$ for the instance in Figure~\ref{fig-semi}.
        Column 1 shows $\mathcal{F}_{k,l}^{vu}$ with different values of $k,l,v$ and $u$; 
        Columns 2,3 and 4 display the semi-feasible packing $\mathcal{D} = \mathcal{P} \cup \mathcal{C} \cup \{D\}$;
        Column 5 displays the corresponding vertex subset $V(\mathcal{D})$.
    }\label{tab-semi}
\end{table}

We can see that a semi-feasible path-cycle packing $\mathcal{D} = \mathcal{P} \cup \mathcal{C} \cup \{D\}$ is feasible if and only if $D$ is a feasible path or cycle.
To obtain the ultimate path-cycle packing,  we will extend the packing by either increasing the length of the path $D$ by 1 or moving $D$ into  $\mathcal{P}\cup \mathcal{C}$ and setting a new edge as $D$. This operation varies depending on the element $D$. If $D$ is a path (but not a feasible path), we will add one edge to it such that it will either become a feasible cycle or a longer path. If $D$ is a feasible cycle, we move $D$ to $\mathcal{C}$ and take a new edge (disjoint with any element in $\mathcal{D}$) as $D$. If $D$ is a feasible path, we can either increase the length of $D$ by adding an edge at its end or move $D$ to $\mathcal{P}$ and take a new edge (disjoint with any element in $\mathcal{D}$) as $D$.


For a semi-feasible packing $\mathcal{D} = \mathcal{P} \cup \mathcal{C} \cup D$, let $V(\mathcal{D}) = V(\mathcal{P}) \cup V(\mathcal{C}) \cup V(D)$.
Let $D_l^{vu}$ denote a path from $v$ to $u$ of length $l$ if $v\neq u$ and a cycle containing $v$ of length $l$ if $u=v$.
We define 
$$\mathcal{F}_{k,l}^{vu} = \{V(\mathcal{D}) \mid  \mathcal{D} = \mathcal{P} \cup \mathcal{C} \cup \{D_l^{vu}\} \wedge c(\mathcal{D})=k\}.$$
Table~\ref{tab-semi} provides some illustrations of $\mathcal{F}_{k,l}^{vu}$ for the example in Figure 1.
Our dynamic programming algorithm indeed stores the vertex sets of all possible semi-feasible path-cycle packings by computing $\mathcal{F}_{k,l}^{vu}$ for all $1 \le k \le 2t$, $v \in V$, $u \in V\setminus B$ and $1 \le l \le \max(l_p,l_c)$.
The following Lemma~\ref{lem-semi-kep} says that we can solve KEP by checking $\mathcal{F}_{k,l}^{vu}$.

\begin{lemma}
    A KEP instance $(G,B,l_p,l_c,t)$ with  $l_p,l_c < t$ is a yes-instance if and only if there exists a tuple ($k$, $v$, $u$, $l$) such that $t \le k \le 2t$, $\mathcal{F}_{k,l}^{vu} \neq \emptyset $, and one of the following two conditions holds:
    \begin{itemize}
        \item  $v \in B$, $u \in V \setminus B$ and $l \le l_p$;
        \item  $v=u \in V \setminus B$ and $l \le l_c$.
    \end{itemize}
    \label{lem-semi-kep}
\end{lemma}

\begin{proof}
    We first prove the forward direction. By Lemma~\ref{lem-2t}, we know that if $(G,B,l_p,l_c,t)$ is a yes-instance, there must be a feasible path-cycle packing $\mathcal{D}$ such that $t \le c(\mathcal{D}) \le 2t$. Since feasible path-cycle packing is also semi-feasible, we know that there exists $\mathcal{F}_{k,l}^{vu}$ satisfying the above conditions that contains $V(\mathcal{D})$.
    
    We consider the other direction.  Now we have $\mathcal{D}=(\mathcal{P},\mathcal{C},D_{l}^{vu})$ with $c(\mathcal{D}) \ge t$. Since it satisfies the above conditions, we know that $D_{l}^{vu}$ is actually a feasible path or cycle. Thus $\mathcal{D}$ is a feasible path-cycle packing with $c(\mathcal{D}) \ge t$, which means $(G,B,l_p,l_c,t)$ is a yes-instance.
\end{proof}

We have the following Lemma~\ref{lem-semi-alg} to compute $\mathcal{F}_{k,l}^{vu}$.

\begin{lemma}
    There exists a dynamic programming algorithm that can compute $\mathcal{F}_{k,l}^{vu}$ for all $1 \le k \le 2t$, $v \in V$, $u \in V\setminus B$ and $l \le \max(l_c,l_p)$ in time $O^*(2^n)$.
    \label{lem-semi-alg}
\end{lemma}

\begin{proof}
Recall that our algorithm is a dynamic programming algorithm. We show how to compute $\mathcal{F}_{k,l}^{vu}$ for different values of $k,l,v$ and $u$ by either setting an initial value or giving a transition equation. 

We distinguish five cases. In fact, if it is not of the first four cases, then there is no packing satisfying the condition of $\mathcal{F}_{k,l}^{vu}$ and we will simply let $\mathcal{F}_{k,l}^{vu}=\emptyset$. This is Case (v).
    \paragraph{Case (i):}$k=l=1$, and $(v,u) \in A(G)$.
    
    If $\mathcal{D} = \{\mathcal{P},\mathcal{C},D_{1}^{vu}\}$ is a semi-feasible path-cycle packing in $\mathcal{F}_{1,1}^{vu}$, we know that $\mathcal{P}=\mathcal{C}=\emptyset$. Notice that  $v\neq u$ since there are no self-loops in $G$ and $(v,u) \in G$. The only possible set in $\mathcal{F}_{1,1}^{vu}$ is $\{v,u\}$. Therefore, we let $\mathcal{F}_{1,1}^{vu} =\{\{v,u\}\}$.
    
    \paragraph{Case (ii):}$k>l=1$, and $(v,u) \in A(G)$.
    
    In this case, for any semi-feasible packing $\mathcal{D} = \{\mathcal{P},\mathcal{C},D_1^{vu}\}$ such that $V(\mathcal{D}) \in \mathcal{F}_{k,1}^{vu}$, it holds that $\mathcal{P} \cup \mathcal{C}\neq \emptyset$ since $k>l$. We assume that $D_{l'}^{v'u'}$ is an arbitrary feasible path or cycle in $\mathcal{P} \cup \mathcal{C}$ (let $v'=u'$ if $D_{l'}^{v'u'}$ is a cycle, otherwise $v' \neq u'$). Then the vertex subset of $\mathcal{D}' = \{\mathcal{P} \setminus \{D_{l'}^{v'u'}\}, \mathcal{C} \setminus  \{D_{l'}^{v'u'}\},D_{l'}^{v'u'}\}$ must be contained in $\mathcal{F}_{k-1,l'}^{v'u'}$. Notice that $v \neq u$ since $l=1$. Thus, we can compute $\mathcal{F}_{k,1}^{vu}$ by enumerating all possible $v'$, $u'$ and $l'$:
    \begin{equation}
        \begin{aligned}
        \mathcal{F}_{k,1}^{vu} = \bigcup_{v' \in B,u' \notin B, \ 1 \le l' \le l_p} \mathcal{F}_{k-1,l'}^{v'u'} \uplus \{v,u\} \ \ \ 
        \cup \\ \bigcup_{v'\notin B, \ 1 \le l' \le l_c} \mathcal{F}_{k-1,l'}^{v'v'} \uplus \{v,u\}.
        \end{aligned}
    \end{equation}
    
    \paragraph{Case (iii):}$k\geq l>1$, $v \neq u$, $l \le l_p$ if $v \in B$, and $l \le l_c$ if $v \notin B$.
   
    In this case, for any semi-feasible packing $\mathcal{D} = \{\mathcal{P},\mathcal{C},D_l^{vu}\}$ such that $V(\mathcal{D}) \in \mathcal{F}_{k,l}^{vu}$, we have that $\mathcal{D}_{l}^{vu}$ is a path of length at least 2 if $v\neq u$ or a cycle if $v=u$. Let $w$ be the predecessor vertex of $u$ and $D_{l-1}^{vw}$ be the path obtained by removing vertex $u$ from $D_l^{vu}$. Then the vertex subset of the semi-feasible packing $\mathcal{D'} = \{\mathcal{P},\mathcal{C},D_{l-1}^{vw}\}$ must be contained in $\mathcal{F}_{k-1,l-1}^{vw}$.
    Thus, we can compute $\mathcal{F}_{k,l}^{vu}$ by enumerating all possible predecessor vertices $w$:
    \begin{equation}
        \begin{aligned}
        \mathcal{F}_{k,l}^{vu} = \bigcup_{w \neq v, \ (w,u) \in A(G)} \mathcal{F}_{k-1,l-1}^{vw} \uplus \{u\}
        \end{aligned}.
    \end{equation}
   
    \paragraph{Case (iv):} $k\geq l>1$, $v = u \notin  B$, and $l \le l_c$.
    
    This case is similar to Case (iii) and we also handle it by enumerating all possible predecessors of vertices $v$:
    \begin{equation}
        \begin{aligned}
        \mathcal{F}_{k,l}^{vv} = \bigcup_{w \neq v, \ (w,v) \in A(G)} \mathcal{F}_{k-1,l-1}^{vw}
        \end{aligned}.
    \end{equation}
    
    \paragraph{Case (v):} None of the above four cases holds. It is easy to verify that there is no corresponding packing satisfying the condition. We simply let $\mathcal{F}_{k,l}^{vu}=\emptyset$.
  
Next, we analyze the running time bound of the algorithm.
Note that our running time bound is not directly obtained  by evaluating the number of $\mathcal{F}_{k,l}^{vu}$ and the time to compute each $\mathcal{F}_{k,l}^{vu}$.
We need a refined analysis.
Note that each vertex subset in some $\mathcal{F}_{k,l}^{vu}$ can be computed in polynomial time. We only need to evaluate how many vertex subsets are contained in $\mathcal{F}_{k,l}^{vu}$.
    Let $num(k,l,v,u)$ denote the number of vertex subsets in $\mathcal{F}_{k,l}^{vu}$ for values of $k,l,v$ and $u$. We have that
    \[
    \begin{aligned}
        &\sum_{k=1}^{\min(2t,n)}\sum_{l=1}^{\max(l_p,l_c) } \sum_{v} \sum_{u} num(k,l,v,u)\\
        &\le \sum_{k=1}^{\min(2t,n)}\sum_{l=1}^{\max(l_p,l_c) } \sum_{v} \sum_{u} \binom{n}{k} \\
        &\le n^3 \sum_{k=1}^{n} {n \choose k}       = O^*(2^n).
    \end{aligned}
    \]
    The number of vertex subsets is bounded by $O^*(2^n)$ and the running time of the algorithm is bounded by $O^*(2^n)$.
\end{proof}

Now we are ready to prove Theorem~\ref{thm-dp}.
\begin{proof}{(of Theorem~\ref{thm-dp})}.
    We prove Theorem~\ref{thm-dp} by showing that KEPalg can determine where there exists a feasible path-cycle packing of length at least $t$ in time $O^*(2^n)$.
    
    There are two mains parts in KEPalg.
    In the first part, we compute $\mathcal{F}_{k,l}^{vu}$ for all $1 \le k \le 2t$, $1 \le l \le \max(l_p,l_c)$, $v \in V$ and $u \in V \setminus B$.
    By Lemma~\ref{lem-semi-alg}, we know that the computation can be done in time $O^*(2^n)$.
    In the second part, we answer the instance by checking $\mathcal{F}_{k,l}^{vu}$.
    If there exists a tuple $(k,l,v,u)$ such that $\mathcal{F}_{k,l}^{vu}$ is nonempty and either the two conditions in Lemma~\ref{lem-semi-kep} holds, the answer of the instance is yes. The checking can be done without using more time than part 1.
    
    The algorithm only answer the decision version of the problem. It can also output the corresponding feasible path-cycle packing by adding an additional data structure to store the corresponding path-cycle packings at each state. 
\end{proof}
    



\subsection{Representative Sets for Kidney Exchange}
\label{sec-rep}

In this Section, our goal is to obtain a fast algorithm based on KEPalg in Theorem~\ref{thm-dp}.
The fast algorithm in this subsection is denoted by KEPalg-rep.

Observe that the exponential factor of the time complexity of KEPalg comes from the cardinality of $\mathcal{F}_{k,l}^{vu}$. Since each collection $\mathcal{F}_{k,l}^{vu}$ might contain at most $\binom{n}{k}$ vertex subsets, the time complexity would roughly be
\[\sum_{k=1}^{\min(2t,n)}\binom{n}{k} \le \sum_{k=1}^{n} \binom{n}{k} = O(2^n).\]

Notice that we already bound the length of feasible path-cycle packings of KEP in Lemma~\ref{lem-2t}. However, the cardinality of $\mathcal{F}_{k,l}^{vu}$ can only be bounded by $\binom{n}{k}$, which is not tight for our analysis above.
This motivates us that we can improve the time complexity by compressing $\mathcal{F}_{k,l}^{vu}$, which might give us a time complexity bounded by $t$. For example, if we can compress $\mathcal{F}_{k,l}^{vu}$ to make it contain at most $\binom{2t}{k}$ subsets, then the time complexity would be reduced to 
$\sum_{k=1}^{2t}\binom{2t}{k} \ = O(4^t).$

To compress $\mathcal{F}_{k,l}^{vu}$, our main technique is the efficient computation of representative sets in~\cite{DBLP:journals/jacm/FominLPS16-rep}.
We first slightly modify the definition of states in our dynamic programming.
Recall that $D_l^{vu}$ denote a path from $v$ to $u$ of length $l$ if $v\neq u$ and a cycle containing $v$ of length $l$ if $u=v$.
At each state, We define a $p$-family $\mathcal{F}_{k,l,p}^{vu} = \{ V(\mathcal{D}) \mid \mathcal{D}=\mathcal{P} \cup \mathcal{C} \cup D_l^{vu} \wedge c(\mathcal{D}) = k \wedge |V(\mathcal{D})|=p\}$.
Note that we can compute $\mathcal{F}_{k,l,p}^{vu}$ for $1 \le k,p \le 2t$ similar to the computation of $\mathcal{F}_{k,l}^{vu}$ in Lemma~\ref{lem-semi-alg}.

Let $q=2t-p$ and $\hat{\mathcal{F}}_{k,l,p}^{vu}$ denote a $q$-representative set of $\mathcal{F}_{k,l,p}^{vu}$.
In our refined algorithm KEPalg-rep, we want to compute $\hat{\mathcal{F}}_{k,l,p}^{vu}$ for $1 \le k,p \le 2t$.
For the convenience, we use another $p$-family $\mathcal{N}_{k,l,p}^{vu}$ as an intermediate variable to compute $\hat{\mathcal{F}}_{k,l,p}^{vu}$.
To compute $\hat{\mathcal{F}}_{k,l,p}^{vu}$, we first compute $\mathcal{N}_{k,l,p}^{vu}$ by following the equation similar to Lemma~\ref{lem-semi-alg} and then compute $\hat{\mathcal{F}}_{k,l,p}^{vu} \subseteq_{rep}^{2t-p} \mathcal{N}_{k,l,p}^{vu}$ by Lemma~\ref{lem-rep-family}.

We recall the five cases in the proof of Lemma~\ref{lem-semi-alg} and give the initialization and transition equation of $\mathcal{N}_{k,l,p}^{vu}$ and $\hat{\mathcal{F}}_{k,l,p}^{vu}$ directly as follows:

\begin{equation}
        \hat{\mathcal{F}}_{k,l,p}^{vu} \subseteq_{rep}^{2t-p} \mathcal{N}_{k,l,p}^{vu} \\
        \label{eq-F-N};
\end{equation}

\begin{equation}
    \begin{aligned}
        &\mathcal{N}_{k,l,p}^{vu} = \\
        &\begin{cases}
              \{\{v,u\}\} & \text{Case(i)} \\
            \bigcup_{v' \in B,u' \in V \setminus B, \ 1 \le l' \le l_p} \hat{\mathcal{F}}_{k-1,l',p-2}^{v'u'} \uplus \{v,u\} &\\
        \cup  \bigcup_{v'\in V \setminus B, \ 1 \le l' \le l_c} \hat{\mathcal{F}}_{k-1,l',p-2}^{v'v'} \uplus \{v,u\} &\text{Case(ii)} \\
            \bigcup_{w \neq v, \ (w,u) \in A} \hat{\mathcal{F}}_{k-1,l-1,p-1}^{vw} \uplus \{u\} &\text{Case(iii)} \\
            \bigcup_{w \neq v, \ (w,v) \in A} \hat{\mathcal{F}}_{k-1,l-1,p}^{vw}  &\text{Case(iv)} \\
              \emptyset & \text{Case(v)}\\
        \end{cases}
    \end{aligned}
        \label{eq-N-F-rep}
\end{equation}



The following Lemma~\ref{lem-rep-F} says that the computation in Equations~(\ref{eq-F-N}) and~(\ref{eq-N-F-rep}) gives us a correct representative set of $\mathcal{F}_{k,l,p}^{vu}$.

\begin{lemma}
    $\hat{\mathcal{F}}_{k,l,p}^{vu} \subseteq_{rep}^{2t-p} \mathcal{F}_{k,l,p}^{vu}$.
    \label{lem-rep-F}
\end{lemma}

\begin{proof}
    We recall that the representative relation is transitive by Lemma~\ref{lem-rep-trans}. By Equation~\ref{eq-F-N}, we know that $\hat{\mathcal{F}}_{k,l,p}^{vu} \subseteq_{rep}^{2t-p} \mathcal{N}_{k,l,p}^{vu}$. If it holds that $ \mathcal{N}_{k,l,p}^{vu} \subseteq_{rep}^{2t-p} \mathcal{F}_{k,l,p}^{vu}$, then we are done because $        \hat{\mathcal{F}}_{k,l,p}^{vu} \subseteq_{rep}^{2t-p} \mathcal{N}_{k,l,p}^{vu} \subseteq_{rep}^{2t-p} \mathcal{F}_{k,l,p}^{vu}.
    $
    Therefore, we only need to prove that $ \mathcal{N}_{k,l,p}^{vu} \subseteq_{rep}^{2t-p} \mathcal{F}_{k,l,p}^{vu}$.

    We prove $ \mathcal{N}_{k,l,p}^{vu} \subseteq_{rep}^{2t-p} \mathcal{F}_{k,l,p}^{vu}$ by induction. 
    
    For $k=1$, by the initialization equation, it holds that $\mathcal{N}_{k,l,p}^{vu}=\mathcal{F}_{k,l,p}^{vu}$ and $\mathcal{N}_{k,l,p}^{vu} \subseteq_{rep}^{2t-p} \mathcal{F}_{k,l,p}^{vu}$. 
    
    By induction, we assume that 
    $$ \mathcal{N}_{i,l,p}^{vu} \subseteq_{rep}^{2t-p} \mathcal{F}_{i,l,p}^{vu}~~\mbox{for}~~ i=1\cdots k-1 \  (k>1).$$ Next, we prove that $\mathcal{N}_{k,l,p}^{vu} \subseteq_{rep}^{2t-p} \mathcal{F}_{k,l,p}^{vu}$.
    
    Without loss of generality, we take Case (iii) as an example and other cases can be handled similarly. For any vertex subset $X \in \mathcal{F}_{k,l,p}^{vu}$ in Case (iii), there must be some family $\mathcal{F}_{k-1,l-1,p-1}^{vw}$ containing $X \setminus \{u\}$. 
    
    Suppose that $Y$ is disjoint with $X$ and $|Y| \le 2t-p$. Then $Y \cup \{u\}$ is disjoint with $X \setminus \{u\}$ and $|Y \cup \{u\}| \le 2t-p+1$. 
    By the induction, we have that 
    $$\hat{\mathcal{F}}_{k-1,l-1,p-1}^{vw} \subseteq_{rep}^{2t-(p-1)} \mathcal{F}_{k-1,l-1,p-1}^{vw}.$$ 
    Thus, there must be a vertex subset $\hat{X} \in \hat{\mathcal{F}}_{k-1,l-1,p-1}^{vw}$ such that $\hat{X}$ is disjoint with $Y \cup \{u\}$. By Equation~\ref{eq-N-F-rep}, we know that $\hat{X} \cup \{u\} \in \mathcal{N}_{k,l,p}^{vu}$. Let $\hat{X}' = \hat{X} \cup \{u\}$. Since $\hat{X}$ is disjoint with $Y$, we have that for any $X \in \mathcal{F}_{k,l,p}^{vu}$, if $|X \cap Y|=\emptyset$ and $|Y| \le 2t-p$, then there exists $\hat X' \in \mathcal{N}_{k,l,p}^{vu}$ such that $|\hat{X}' \cap Y|=\emptyset$. 
    
    Therefore, we have proven that $ \mathcal{N}_{k,l,p}^{vu} \subseteq_{rep}^{2t-p} \mathcal{F}_{k,l,p}^{vu}$.
\end{proof}

The following Lemma~\ref{lem-semi-kep-rep} indicates that we can solve KEP by checking $\hat{\mathcal{F}}_{k,l,p}^{vu}$.

\begin{lemma}
    Suppose that we are given a KEP instance $(G,B,l_p,l_c,t)$. If $l_p,l_c < t$, then $(G,B,l_p,l_c,t)$ is a yes-instance if and only if there exists a tuple ($k$, $l$, $p$, $v$, $u$) such that $\hat{\mathcal{F}}_{k,l,p}^{vu} \neq \emptyset $ and either of the two conditions holds:
    \begin{itemize}
        \item $t \le k  \le p \le 2t$, $v \in B$, $u \in V \setminus B$ and $l \le l_p$,
        \item $t \le k \le p \le 2t$, $v=u \in V \setminus B$ and $l \le l_c$.
    \end{itemize}
    \label{lem-semi-kep-rep}
\end{lemma}

\begin{proof}
    The proof is based on the proof of Lemma~\ref{lem-semi-kep}. We first prove the backward direction. By the definition, we have \[
        \hat{\mathcal{F}}_{k,l,p}^{vu} \subseteq \mathcal{F}_{k,l,p}^{vu} \subseteq \mathcal{F}_{k,l}^{vu}.
    \]
    This implies that if $\hat{\mathcal{F}}_{k,l,p}^{vu} \neq \emptyset$, then $\mathcal{F}_{k,l}^{vu} \neq \emptyset$. Note that if a tuple $(k,l,p,v,u)$ satisfies either the conditions in Lemma~\ref{lem-semi-kep-rep}, then the tuple $(k,l,v,u)$ will naturally satisfies either the conditions in Lemma~\ref{lem-semi-kep}. Thus by Lemma~\ref{lem-semi-kep}, we have that $(G,B,l_p,l_c,t)$ is a yes-instance.
    
    Next, we consider the forward direction.
    
    Suppose that $(G,B,l_p,l_c,t)$ is a yes-instance.
    There exists a vertex subset with the smallest size $p_0$ in all $\mathcal{F}_{k,l}^{vu}$ satisfying either the two conditions in Lemma~\ref{lem-semi-kep}.
    We assume that $X_0 \in \mathcal{F}_{k_0,l_0}^{v_0u_0}$ without loss of generality.
    Since $X \in \mathcal{F}_{k_0,l_0,p_0}^{v_0u_0}$, we have that $\hat{\mathcal{F}}_{k_0,l_0,p_0}^{v_0u_0} \neq \emptyset$ by the definition of representative sets.
    Then it is enough to prove that $t \le p_0 \le 2t$. Note that $p_0 \ge k_0 \ge t$. Thus we only need to prove that $p_0 \le 2t$.
    
    Let $\mathcal{D}_0$ be the corresponding semi-feasible path-cycle packing of $X_0$.
    If there are only cycles in $\mathcal{D}_0$, we have that $p_0=k_0 \le 2t$.
    Otherwise, we can assume that there exists a path $P$ in $\mathcal{D}_0$.
    Since there are no isolated vertices in $\mathcal{D}_0$, it holds that $p_0 \le 2k_0$. If $k_0 = t$, we have that $p_0 \le 2k_0 = 2t$. Thus, we only need to consider the case $k_0 > t$.
    
    We construct another semi-feasible path-cycle packing $D'$ by removing the endpoint of $P$ (If P is a path of length 2, we remove it from $\mathcal{D}$ directly).
    Without loss of generality, we assume that $V(\mathcal{D'}) \in \mathcal{F}_{k',l'}^{v'u'}$.
    Note that we have $k' = k_0 -1 \ge t$.
    The tuple $(k',l',v',u')$ with $\mathcal{F}_{k',l'}^{v'u'} \neq \emptyset$ satisfies either the two conditions in Lemma~\ref{lem-semi-kep}.
    However, we have that $|V(\mathcal{D}')| < |X_0|$, which is a contradiction to the minimality of $X_0$. So, we know that it holds $p_0 \le 2t$.    
\end{proof}

We use the following Lemma~\ref{lem-semi-alg-kep} to compute $\mathcal{F}_{k,l}^{vu}$. 

\begin{lemma}
    There exists a dynamic programming algorithm that can compute $\mathcal{N}_{k,l,p}^{vu}$ and $\hat{\mathcal{F}}_{k,l,p}^{vu}$ for all $1 \le k \le p \le 2t$, $v \in V$, $u \in V \setminus B$ and $l \le \max(l_c,l_p)$ in $O^*(6.855^{t})$ time.
    \label{lem-semi-alg-kep}
\end{lemma}

\begin{proof}
    First, notice that the bottleneck of the time complexity in Lemma~\ref{lem-semi-alg-kep} is the computation of $\hat{\mathcal{F}}_{k,l,p}^{vu}$, i.e., the computation of the representative set of $\mathcal{N}_{k,l,p}^{vu}$ by Lemma~\ref{lem-rep-family}.
    Let $$q=2t-p ~~~\mbox{and}~~~ s_{p,q} = x^{-p}(1-x)^{-q}\cdot2^{o(p+q)}.$$ 
    By Lemma~\ref{lem-rep-family}, we have that $|\hat{\mathcal{F}}_{k,l,p}^{vu}| \le s_{p,q}$.  
    
    We set $x=\frac{p}{p+2q}$.   
    To simplify our analysis, we
    first prove the following purely computational result
    \begin{equation}
    s_{p,q} \le e^2 (p+1) \cdot s_{p+1,q-1} \cdot 2^{o(p+q)}.
    \label{eq-998}
    \end{equation}
 Note that $(1+\frac{1}{t})^t < e$ for all $t>0$. We have that
    \[
        \begin{aligned}
            &\frac{s_{p,q}}{s_{p+1,q-1}} = \frac{{(\frac{p}{p+2q})}^{-p}(1-\frac{p}{p+2q})^{-q}}{{(\frac{p+1}{p+2q-1})}^{-(p+1)}(1-\frac{p+1}{p+2q-1})^{-(q-1)}}\\
            &= \big( 1+\frac{1}{p+2q-1} \big) ^{p+q} \cdot \frac{(p+1)^{p+1}}{p^p} \cdot \frac{(2q-2)^{q-1}}{(2q)^q}\\
            &\le \big( 1+\frac{1}{p+2q-1} \big) ^{p+2q-1} \cdot (1+\frac{1}{p})^p (p+1) \\
            &\le e^2(p+1).
        \end{aligned}
    \]
    Thus, we get (\ref{eq-998}).

    By Equation~(\ref{eq-N-F-rep}), we have that 
    \[
        \begin{aligned}
            &|\mathcal{N}_{k,l,p}^{vu}| \le max\{n^2\cdot 2t s_{p-2,q+2} , \ n s_{p-1,q+1}, ns_{p,q}\} \\ 
            &\le 2e^4n^2t(p-1)p \cdot s_{p,q} \\
            &\le 8e^4n^2t^3\cdot s_{p,q}.
        \end{aligned}
    \]
    By Lemma~\ref{lem-rep-family}, the time of computing all $\hat{\mathcal{F}}_{k,l,p}^{vu}$ can be bounded as follows: 
    \[
        \begin{aligned}
            &\sum_{k=2}^{2t}\sum_{l=1}^{\max(l_c,l_p)}\sum_{v,u}\sum_{p=2}^{2t} |\mathcal{N}_{k,l,p}^{vu}| \cdot (1-x)^{-q} \cdot 2^{o(p+q)} \cdot \log n \\
            &\le \sum_{k=2}^{2t}\sum_{l=1}^{t}\sum_{v,u}\sum_{p=2}^{2t} 8e^4n^2t^3 \cdot s_{p,q} \cdot (1-x)^{-q} \cdot 2^{o(p+q)} \cdot \log n \\
            &= \sum_{k=2}^{2t}\sum_{l=1}^{t}\sum_{v,u} 8e^4n^2t^3\cdot 2^{o(p+q)} \cdot  \log n  \sum_{p=2}^{2t}s_{p,q} \cdot (1-x)^{-q}  \\
            &= 16e^4n^4t^5 \cdot 2^{o(p+q)} \cdot \log n \sum_{p=2}^{2t}  s_{p,q}\cdot (1-x)^{-q}  \\
            &\le 32e^4n^4t^6 \cdot 2^{o(p+q)} \cdot \log n \max_{p=2}^{2t} s_{p,q}\cdot (1-x)^{-q}.
        \end{aligned}
    \]
    Recall that $x = \frac{p}{p+2q} = \frac{p}{4t-p}$. We get that
    \[
        \begin{aligned}
            & \max_{p=2}^{2t}  \ s_{p,q}\cdot (1-x)^{-q} \\
            &= \max_{p=2}^{2t} \  x^{-p}(1-x)^{-2q}\cdot2^{o(p+q)} \\
            &= 2^{o(p+q)}\max_{p=2}^{2t}  {(\frac{p}{p+2q})}^{-p}{(1-\frac{p}{p+2q})}^{-2q}\\
            &= 2^{o(t)}\max_{p=2}^{2t-1}  {(\frac{4t-p}{p})}^{p}{(\frac{4t-p}{4t-2p})}^{4t-2p}.
        \end{aligned}
    \]
    We can assume that $p=\alpha t $ ($\alpha \in (0,2)$) and get
    \[
        \begin{aligned}
            &\max_{p=2}^{2t-1}  {(\frac{4t-p}{p})}^{p}{(\frac{4t-p}{4t-2p})}^{4t-2p}\\
            &=\max_{0 < \alpha< 2}^{}  [{(\frac{4-\alpha}{\alpha})}^{\alpha}{(\frac{4-\alpha}{4-2\alpha})}^{4-2\alpha}]^t.
        \end{aligned}
    \]
    To obtain the time complexity, our goal is to optimize function
    $
        f(\alpha) = {(\frac{4-\alpha}{\alpha})}^{\alpha}{(\frac{4-\alpha}{4-2\alpha})}^{4-2\alpha}
    $ when $0 < \alpha < 2$.
    By direct calculation, we know that $f(\alpha)$ get the maximum value when $\alpha_0 = 2-\frac{2}{5}\sqrt{5}$ and $f(\alpha_0) < 6.855$.
    Therefore, we can compute $\mathcal{N}_{k,l,p}^{vu}$ and $\hat{\mathcal{F}}_{k,l,p}^{vu}$ in time $O^*(6.855^{t})$.
\end{proof}

    

Finally, we are ready to prove Theorem~\ref{thm-kep}.
\begin{proof}{[\textbf{Proof of Theorem~\ref{thm-kep}}]}
    We prove Theorem~\ref{thm-kep} by showing that the algorithm \textsc{KEPalg-rep} can determine whether there exists a feasible path-cycle packing of total length at least~$t$ in time $O^*(6.855^t)$.
The algorithm \textsc{KEPalg-rep} consists of three main components:

\noindent \textbf{Part 1: Handling Large Paths or Cycles.}
We first check whether $l_c \ge t$ or $l_p \ge t$. In this case, if there exists a feasible cycle (resp., path) of length at least~$t$ when $l_c \ge t$ (resp., $l_p \ge t$), we return \textsc{Yes} and output the corresponding cycle or path. Otherwise, we safely set $l_c = \min(l_c, t-1)$ and $l_p = \min(l_p, t-1)$. To compute the feasible paths and cycles, we use the algorithm from~\cite{DBLP:journals/jacm/FominLPS16-rep}, which runs in time $O^*(6.75^{t+o(t)})$.

\noindent \textbf{Part 2: Computing Table Entries.}
Next, we compute the sets $\mathcal{N}_{k,\ell,p}^{vu}$ and $\hat{\mathcal{F}}_{k,\ell,p}^{vu}$ for all $1 \le k \le p \le 2t$, $1 \le \ell \le \max(l_p, l_c)$, $v \in V$, and $u \in V \setminus B$. According to Lemma~\ref{lem-semi-alg-kep}, this step can be carried out in time $O^*(6.855^t)$.

\noindent \textbf{Part 3: Answering the Instance.}
Finally, we determine whether the instance is \textsc{Yes} by inspecting the entries $\hat{\mathcal{F}}_{k,\ell,p}^{vu}$. If there exists a tuple $(k,\ell,p,v,u)$ such that $\hat{\mathcal{F}}_{k,\ell,p}^{vu}$ is nonempty and at least one of the two conditions in Lemma~\ref{lem-semi-kep-rep} holds, we return \textsc{Yes}; otherwise, we return \textsc{No}. This check can be performed within the same time bounds as Part~2.

Similar to \textsc{KEPalg}, the algorithm \textsc{KEPalg-rep} can be extended to output a corresponding feasible path-cycle packing by maintaining an additional data structure during the computation.
\end{proof}


\section{Randomized Algorithm via Color Coding}

In this section, we formally show that, by applying Lemma~\ref{lem-2t}, we can directly improve the running time of \citeauthor{DBLP:conf/ijcai/MaitiD22}'s randomized algorithm to $O^*(4^t)$. This matches the best known running time for randomized algorithms, as achieved in~\cite{DBLP:conf/ijcai/Hebert-JohnsonL24}.

\label{sec-random}

\begin{theorem}
    KEP can be randomly solved via Color Coding in $O^*(4^t)$ time.
    \label{thm-3}
\end{theorem}
\begin{proof}
    We follow the structure of the proof of Theorem 1 in~\cite{DBLP:conf/ijcai/MaitiD22}.
    
    Let $(G,B,l_p,l_c,t)$ be an arbitrary instance of KEP.
    We first handle the case when $l_p \ge t$ and $l_c \ge t$.
    For this case, we check whether there is a feasible path of length $t$ or a feasible cycle of length at least $t$. If so, we find a solution covering at least $t$ recipients and the problem is solved.
    If no, we can set that $l_p = \min(l_p, t-1)$ and $l_c = \min(l_c,t-1)$.
    By slightly modifying the randomized algorithms of $k$-Path~\cite{DBLP:books/sp/CyganFKLMPPS15-ran-path} and Long Directed Cycle~\cite{DBLP:journals/ipl/Zehavi16-ran-cycle}, finding feasible paths and cycles of length at least $t$ can be done in time $O^*(2^t)$ and $O^*(4^t)$, respectively.

    Then we can assume that $l_p < t$ and $l_c < t$. By Lemma~\ref{lem-2t}, we have that $(G,B,l_p,l_c,t)$ is a yes-instance if and only if there exists a path-cycle packing $\mathcal{D}$ such that $t \le c(\mathcal{D}) \le 2t$.
    We color each vertices in $G$ uniformly at random from a set $S$ of $2t$ colors. We say that a color $\chi:V(G) \to [2t]$ is proper if every vertex in $\mathcal{D}$ gets a different color. Then for a yes-instance $(G,B,l_p,l_c,t)$, the graph $G$ is colored properly with probability at least \[\frac{(2t)!}{(2t)^{2t}} \ge e^{2t}.\]

    We then prove that if $G$ is colored properly, we can find $\mathcal{D}$ in $O^*(4^t)$ time. We let $f(X,v) = 1$ if there exists a feasible path $P$ ending at $v$ such that $|V(P)| = |X|$ and $\bigcup_{v \in V(P)}\{\chi(v)\} = X$ . Otherwise we let $f(X,v) = 0$. 
    For the initialization equation, we have
    \[
    f(X, v) = 
    \begin{cases} 
        0, &  (\chi(v) \notin X) \lor (|X| = 1 \land v \notin B), \\ 
        1, &  (\chi(v) \in X) \land (|X| = 1) \land (v \in B).
    \end{cases}
    \]
    For the transformation equation, we have that 
    \[
    f(X,v) = \bigvee_{\substack{u \in V(G) \\ (u,v) \in E(G)}} f(X \setminus \{\chi(v)\}, u).  
    \]

    We can also define $g(X,v,u)$ if there exists a path starting at $v$ and ending at $u$ such that $|V(P)|=|X|$ and  $\bigcup_{v \in V(P)}\{\chi(v)\} = X$. Note that $g(X,v,u)$ can be computed similarly.
    
    We let $\mathcal{F}(X)$ be 1 iff there exists a feasible path/cycle $D$ such that $|V(D)| = |X|$ and $\bigcup_{v \in V(D)}\{\chi(v)\} = X$ and let $\mathcal{G}(X)$ be 1 iff there exists a feasible path-cycle packing $\mathcal{D}$ such that $|V(\mathcal{D})| = |X|$ and $\bigcup_{v \in V(\mathcal{D})}\{\chi(v)\} = X$. By the definition, we have that
    $$
    \begin{aligned}
        \mathcal{F}(X) = ( \bigvee_{v \in V(G) \setminus B}f(X,v)) \bigvee \\ 
        (\bigvee_{v,u \in V(G) \setminus B, (u,v) \in E(G)} g(X,v,u));
    \end{aligned}
    $$
    
    $$
    \begin{aligned}
        \mathcal{G}(X) = \begin{cases}
            0 &X = \emptyset,\\
            \bigvee_{\emptyset \subsetneq Y \subseteq X} \mathcal{F}(Y) \land \mathcal{G}(X \setminus Y) &\text{Otherwise}.
        \end{cases}  
    \end{aligned}
    $$ 
    We then solve KEP by checking whether there exists $\mathcal{G}(X)=1$ such that $t \le |X| \le 2t$.

    By the equations above, it is easy to see that $f(X,v)$, $g(X,v,u)$ and $\mathcal{F}(X)$ can be computed in $O^*(2^{2t})$ time. Similar to the exact algorithms in~\cite{DBLP:conf/ijcai/XiaoW18}, $\mathcal{G}(X)$ can also be computed in $O^*(2^{2t})$ time based on the subset convolution technique~\cite{DBLP:conf/stoc/BjorklundHKK07-sub-con}. Hence, the total running time of the randomized algorithm is $O^*(4^t)$.
    
\end{proof}

\section{Conclusion and Discussion}
In this paper, we study KEP parameterized by the number of covered recipients \(t\) and present a deterministic algorithm with running time \(O^*(6.855^t)\). To achieve this result, we introduce the technique of representative sets into KEP. With modifications to the dynamic programming framework, our approach can be extended to various KEP variants~\cite{DBLP:journals/talg/KrivelevichNSYY07,DBLP:journals/dmaa/BiroMR09} as well as to other path- and cycle-related packing problems.

Several promising directions remain for further improving our algorithm. At present, our procedure computes vertex subsets for all semi-feasible path-cycle packings of length at most \(2t\), as described in Lemmas~\ref{lem-2t} and~\ref{lem-semi-kep}. It may be possible to optimize this step by focusing solely on vertex subsets of semi-feasible packings of length at most \(t\) and then determining in polynomial time whether these can be extended to a full feasible path-cycle packing.

In the realm of randomized algorithms, the best known running time is \(O^*(4^t)\). By applying Lemma~\ref{lem-2t}, one can show that the randomized algorithm of \citeauthor{DBLP:conf/ijcai/MaitiD22}~\cite{DBLP:conf/ijcai/MaitiD22} also achieves a running time of \(O^*(4^t)\). Whether this bound can be further improved remains an open and challenging problem.

Another intriguing avenue for future work is the development of even faster algorithms for KEP when \(l_p\) and \(l_c\) are small constants. This consideration is particularly significant in practical applications, where these parameters are often naturally limited in size.

\section*{Acknowledgments}
We thank the anonymous reviewers for their valuable comments and suggestions that helped improve the quality of this
paper. The work is supported by the National Natural Science
Foundation of China, under the grants 62372095 and 62502078.

\cleardoublepage

\bibliographystyle{named}
\bibliography{ijcai26}

\end{document}